\documentclass[]{article}

\usepackage[leqno]{amsmath}
\usepackage{amssymb,amsfonts,xfrac}
\usepackage{amsthm}
\usepackage{cite}
\usepackage{hyperref}
\usepackage{todonotes}
\usepackage{enumitem}
\usepackage{graphicx}
\usepackage{tikz}
\usepackage{color}
\usepackage{import}
\usepackage{mathtools} 
\usepackage{authblk}

\numberwithin{equation}{section}

\newtheorem{theorem}{Theorem}[section]

\newtheorem{proposition}[theorem]{Proposition}
\newtheorem{lemma}[theorem]{Lemma}

\newtheorem*{theorem*}{Theorem}
\newtheorem*{claim*}{Claim}
\newtheorem*{proposition*}{Proposition}
\newtheorem*{lemma*}{Lemma}
\newtheorem*{corollary*}{Corollary}

\theoremstyle{definition}
\newtheorem{definition}[theorem]{Definition}

\newtheorem{remark}[theorem]{Remark}
\newtheorem{example}[theorem]{Example}

\newtheorem{fact}[theorem]{Fact}

\newtheorem*{definition*}{Definition}
\newtheorem*{observation*}{Observation}
\newtheorem*{remark*}{Remark}
\newtheorem*{example*}{Example}
\newtheorem*{question*}{Question}
\newtheorem*{exercise*}{Exercise}
\newtheorem*{fact*}{Fact}
\newtheorem*{notation*}{Notation}

\renewcommand{\mod}{\bmod\,}
\newcommand{\ZZ}{\mathbb{Z}}
\newcommand{\FF}{\mathbb{F}}

\newcommand{\ii}{^{-1}}

\renewcommand{\epsilon}{\varepsilon}
\newcommand{\lnk}[1]{\href{#1}{\color{blue}{link}}}

\DeclareMathOperator{\SL}{SL}
\DeclareMathOperator{\PSL}{PSL}

\DeclareMathOperator{\PP}{Pr}

\usepackage{pythonhighlight}

\title{Post-quantum hash functions using $\mathrm{SL}_n(\mathbb{F}_p)$}
\author{Corentin Le Coz\thanks{Supported by the Israel Science Foundation (grant no. 2919/19), the FWO and the FNRS under the Excellence of Science (EOS)
program (project ID 40007542)}, Christopher Battarbee, Ram\'{o}n Flores\thanks{Partially supported by grant PID2020-117971GB-C21 of the Spanish Ministry of Science and
Innovation, and grant FQM-213 of the Junta de Andaluc\'{i}a.},\\ Thomas Koberda\thanks{Partially supported by NSF grant DMS-2002596}, and Delaram Kahrobaei\thanks{Partially supported by a
Canada's New Frontiers in Research Fund, under the Exploration grant entitled ``Algebraic Techniques for Quantum Security" as well as a grant from CUNY. DK also has been partially supported by an ONR grant, 62909-24-1-2002}}

\affil{Oppida, France, corentinlecoz@outlook.com\\
Department of Computer Science, University of York, UK, cb2036@york.ac.uk\\
 Department of Geometry and Topology, University of Seville, Spain,
ramonjflores@us.es\\
Department of Mathematics, University of Virginia, thomas.koberda@gmail.com\\
Computer Science and Mathematics Departments, Queens College CUNY; Department of Computer Science, University of York, UK; Initiatives for the Theoretical Sciences (ITS), CUNY Graduate Center,
dkahrobaei@gc.cuny.edu}

\begin{document}
\maketitle

\begin{abstract}
We define new families of Tillich-Z\'emor hash functions, using higher dimensional special linear groups over finite fields as platforms.
The Cayley graphs of these groups combine fast mixing properties and high girth, which
together give rise to good preimage and collision resistance of the
corresponding hash functions. We justify the claim that the resulting
hash functions are post-quantum secure.
\end{abstract}


\section{Introduction}

Hash functions obtained from families of expander graphs were introduced by Charles-Lauter-Goren in \cite{charles2008cryptographic}, where they were in turn
inspired by a scheme of Tillich-Z\'emor \cite{Tillich1994Algebraic}.
Charles-Lauter-Goren considered specific families of expander graphs discovered by
Lubotzky-Phillips-Sarnak \cite{Lubotzky1988Ramanujan} and Pizer \cite{Pizer1990Ramanujan}.
The Charles-Lauter-Goren construction is quite general, and can be applied
to any expander family in which finding cycles is hard, and thereby
furnishes collision resistant hash functions.
Similar schemes have been proposed by several authors; see \cite{Shpilrain2016Compositions,Bromberg2017Navigating,Ghaffari2018More}, and \cite{Sosnovski2018} for a survey of this topic.

Interest in hash functions based on novel platforms
fits into the context of recent
efforts to modernize existing hash functions, and to adapt them
to the design and security of hash-based consensus mechanisms,
most notably with respect to blockchains \cite{Borthwick2020},
and especially in light of the recently proved practicality of finding
collisions in the SHA-1 hashing algorithm \cite{Stevens2017First}.

The general idea behind Tillich-Z\'emor hash functions is the following.
Fixing a base vertex, the input of the hash function is interpreted as a sequence of
instructions, resulting in a non-backtracking path in a $d$-regular graph.
The output of the hash function is the endpoint vertex of the path.
More precisely, the input is a string of numbers in
\[[d-1] \vcentcolon = \{1, 2, \ldots, d-1\}\] of arbitrary
length, and the output is the vertex obtained by performing a simple walk starting at a base vertex, using the elements of $[d-1]$ as transition data for subsequent steps in the
walk. See Definition \ref{definition: general hash function} below for details.

A well-constructed hash function is an efficiently computable function which
enjoys two main features. The first is \emph{preimage resistance},
which means that given a point in the
hash value, it is computationally hard to find an input that maps to that hash value.
The second,  is \emph{collision resistance}, which requires the problem of finding distinct inputs with the same output to be computationally difficult.

The main goal of this paper is to propose a new Goren-Lauter-Charles--type
scheme, where the hash functions use Cayley graphs of the special linear groups $\SL_{n}(\FF_p)$ as platforms, where here $p$ is prime and $n \ge 3$.
The restriction to the fields $\FF_p$ comes from the fact that the corresponding groups will be obtained as quotients of $\SL_n(\ZZ)$.
A crucial observation is that, in schemes using these groups as a platform,
the problem of finding a preimage or a collision
corresponds to finding factorizations of the identity matrix with prescribed factors.
With this observation in hand,
and by taking into account recent work of Arzhantseva-Biswas \cite{Arzhantseva2018Large}
concerning the expanding properties of the Cayley graphs of these groups, we offer a detailed study of the security of our protocol. In particular, we have
the following:

\begin{itemize}

\item \emph{Preimage resistance}. In Proposition \ref{expansion}, we collect
the expansion properties the family of Cayley graphs $\{G_{n,p}\}_{p}$ of the groups $\SL_{n}(\FF_p)$, where
$n\geq 3$ is fixed and where $p$ tends to infinity. Expansion in this family
of graphs guarantees good mixing properties, because under these conditions
the random walk gives a good approximation to the uniform distribution after
$O(\log p)$ steps.

\item \emph{Collision resistance}. The strength of our hash function
with respect to collision resistance is mainly based on the absence of small
cycles in the Cayley graphs of the underlying groups. In fact,
Proposition \ref{proposition: freeness implies big girth} provides a lower bound on
the girth of the graphs $\{G_{n,p}\}_{p}$ on the order of $\log p$. It follows
that a factorization of the identity, which is easily seen to be equivalent to finding
a collision for the hash function, is in turn equivalent to solving over
a system of $n^2$ equations in a number of variables that is $O(\log p)$, over the
field $\mathbb F_p$. In full generality, solving such systems of equations is NP-hard.

Replacing the problem of factoring the identity with the problem of factoring an
arbitrary group element yields a similar system of equations, lending further evidence of resistance of the hash function to finding
preimages; see Section \ref{subsection: resistance via polynomial equations}.

\end{itemize}

For $n=2$, certain Cayley graphs of the groups $\PSL_2(\FF_p)$ give rise to the
celebrated Lubotzky--Phillips--Sarnak expander graphs \cite{Lubotzky1988Ramanujan},
which were then used to build hash functions in \cite{charles2008cryptographic}.
A successful collision attack (i.e.~an efficient computation of a
collision) was found in \cite{TillichZemor08Collision}, by taking coefficients in $\ZZ[i]$ and
then reducing to a system of equations of degree two.
More recently, Sardari \cite{Sardari} attacked preimage resistance by designing a polynomial-time algorithm that represents a number as a sum of four squares with some restricted congruence conditions.
The essentially different nature of higher dimensional special linear groups gives
evidence of additional security, and makes it likely that these attacks are
far more difficult to execute for the hash functions proposed here.

Considering symmetric generating sets enables us to employ results
from the theory of simple random walks in simplicial graphs.
Nevertheless, the fact that we restrict ourselves to non backtracking
random walks precludes the use of multiplicativity of the hash function,
and thus complicates parallel computation.
We discuss these issues in Section \ref{subsection: homomorphism parallel computing}.

\paragraph{\it Structure of the paper:}
In Section \ref{section: definition hash function} we give
the relevant group theoretic background, define the hash functions,
prove the expansion property of the Cayley graphs, and
exhibit concrete examples.
In Section \ref{section: properties hash function}, we describe the various
properties of our scheme: namely, we relate free generation with a lower bound
on the girth; we then describe the role of polynomial equations in
preimage finding and in collision resistance.
Finally, we discuss multiplicativity and parallel computing,
showing that the collision attack from Grassl et. al. \cite{grassl2011cryptanalysis} using palindromes
does not break the scheme presented here.
Section \ref{section: conclusions} concludes the paper.

In Appendix \ref{section: implementation} we provide a Python/Sage implementation of an instance of the hash functions considered in the paper.

\section{Definition of the hash functions}\label{section: definition hash function}

This section defines our hash functions and exhibits concrete instances.
We start by recalling some relevant background material which will be
essential in our construction and in the sequel.

\subsection{Background about special linear groups}

For general results about special linear groups over finite fields,
we refer the reader to Hall's book \cite{Hall2015}. In this section
we concentrate on a number of properties established by Arzhantseva-Biswas
in their article \cite{Arzhantseva2018Large}. We summarize their
results in the following theorem:

\begin{theorem}[Arzhantseva-Biswas \cite{Arzhantseva2018Large}]
\label{theorem: ArzBis}

Let $n\geq 2$ and let $p$ a prime. Write $\pi_p: \SL_n(\ZZ) \to \SL_n(\FF_p)$ for the canonical projection given
by reduction modulo $p$. There exist matrices
$\tilde{A},\tilde{B} \in \SL_n(\ZZ)$ such that:
\begin{enumerate}[label=(\roman*)]
	\item\label{item: Ap and Bp generate} There exists a prime $p_0$ such that
	for all $p \ge p_0$, the matrices $\tilde{A}_p \vcentcolon= \pi_p(\tilde{A})$ and
	$\tilde{B}_p \vcentcolon= \pi_p(\tilde{B})$ generate $\SL_n(\FF_p)$.
	\item\label{item: <A,B> is free} If $\langle \tilde{A},\tilde{B} \rangle$ is the subgroup generated
	by $\tilde{A}$ and $\tilde{B}$ inside of $\SL_n(\ZZ)$, then  $\langle \tilde{A},\tilde{B} \rangle$
	is isomorphic to $F_2$, the free group of rank two.
	\item\label{item: diameter upper bound} The diameter of the Cayley graph
	$G_{n,p}$ of $\SL_n(\FF_p)$ with respect to $\{\tilde{A}_p^{\pm1},\tilde{B}_p^{\pm1}\}$
	is $O(\log p)$.
\end{enumerate}

\end{theorem}

Observe that for $n\geq 3$, items~\ref{item: Ap and Bp generate} and~\ref{item: <A,B> is free}
	 reflect the fact that the subgroup of $\SL_n(\ZZ)$ generated by
	 $\tilde A$ and $\tilde B$ is usually a thin subgroup of $\SL_n(\mathbb R)$.
	 The fact that $\tilde{A}_p$ and $\tilde{B}_p$ generated the corresponding
	 finite quotients for
	all but finitely many values of $p$ is a reflection of strong/superstrong
	approximation. In turn, item \ref{item: diameter upper bound}
	implies that the girth of $G_{n,p}$ is optimal, because the graphs $\{G_{n,p}\}_p$ form a family of expander graphs (see Proposition \ref{expansion} below).

\begin{remark}
When using the Cayley graphs $G_{n,p}$ as a platform, we think of $n$ as being
fixed and $p$
as modulating the level of security, with the trade-off being that the hash functions become more expensive to compute for large $p$.
\end{remark}

Possible choices for $\tilde{A}$ and $\tilde{B}$ are given by:
\[
\tilde A = \begin{pmatrix}
	1 & a & 0 & 0& \ldots & 0\\
	0 & 1 & a & 0& \ldots & 0\\
	0 & 0 & 1 & a & \ldots &0\\
	\vdots & & & & & \vdots\\
	&&&&\ldots&a\\
	0 & 0 & 0& 0&\ldots & 1\\
	\end{pmatrix},\,
	\tilde B = \begin{pmatrix}
	1 & 0 & 0 & \ldots & &0\\
	b & 1 & 0 & \ldots & &0\\
	0 & b & 1 &  \ldots& &0\\
	0 & 0 & b & \ldots & &0\\
	\vdots &&&&\vdots&\vdots\\
		0& 0 & 0 & \ldots & {\strut \strut b} & 1\\
	\end{pmatrix}\in \SL_{n}(\ZZ),
\] with $a,b\geq 2$. These matrices will be crucial in the description of our hash function.

\subsection{Expansion}

For us to implement the Charles-Lauter-Goren approach, we must take advantage of the
good mixing properties of the expander graphs.

\begin{proposition}
\label{expansion}

For $n\geq 3$ fixed and $p\rightarrow\infty$, the sequence
$\{G_{n,p}\}_{p}$ is a family of expander graphs.

\end{proposition}

\begin{proof}[Sketch of proof]

By item \ref{item: Ap and Bp generate} of Theorem\ref{theorem: ArzBis},
there exists a prime $p_0$ such that
	for all $p \ge p_0$, the matrices $\tilde{A}_p \vcentcolon= \pi_p(\tilde{A})$ and
	$\tilde{B}_p \vcentcolon= \pi_p(\tilde{B})$ generate $\SL_n(\FF_p)$.
	Now, since $\SL_{n}(\ZZ)$ has property (T) for $n\ge3$ \cite{bekka2008kazhdan},
	and $\SL_2(\ZZ)$ has property $(\tau)$ with respect to the family of congruence subgroups\cite{Lubotzky2005What}, the proposition follows.
\end{proof}

An immediate consequence of this proposition is that the random walk
approximates the uniform distribution after $O(\log p)$ steps in the
corresponding graph $G_{n,p}$, as we will elaborate in
Section \ref{subsection: mixing and indistinguishability}.
	We note that in \cite{Bromberg2017Navigating}, random walks are conducted on
	Cayley graphs with respect to non-symmetric generating sets, and thus their
	asymptotic
	behavior is less clear.
	Similar issues arise in \cite{Tomkins2020New}, since then hash values could be restricted to a proper subgroup.
	As stated in \cite{Arzhantseva2018Large}, we note that one can effectively compute the lower bound $p_0$.
	No explicit bound on $p_0$ has been given,
	though by combining existing results
	one can probably prove that $p$ need not be very large,
	likely on the order of magnitude of $n$;
	see for instance \cite[Appendix]{Golsefidy2012Expansion} and
	\cite[Theorem D]{Guralnick1999Small}. Note that the larger the value of the prime $p$,
	the more secure the hash function.

\subsection{The general construction}
We now use matrices given in \cite{Arzhantseva2018Large} to define an explicit
family of hash functions.

\begin{definition}[Special linear group based hash functions]\label{definition: general hash function}
Let $n \ge 3$ and let $p$ be a prime number.
Let $a,b,\ell \ge 2$ that satisfy:

\begin{itemize}
	\item If $n=3$, $a \equiv 1 (\mod 3)$, $b \equiv -1 (\mod 3)$ and $\ell=4^k$ for some positive integer $k$.
	\item If $n\geq 4$, there exists a prime $q$ such that $n \equiv a \equiv b \equiv 1 (\mod q)$ and $\ell$ is at least $3(n-1)$ and is of the form $q^{k+1}+1$ for some integer $k$.
\end{itemize}

Consider the matrices $\tilde{A}$ and $\tilde{B}$ from the previous section. In the following we will denote $A\equiv \tilde{A}^{\ell}$ and $B\equiv \tilde{B}^{\ell}$.

We use the notation $[k]$ to denote the set of integers from 1 to $k$, and $[k]^*$ to denote the set of finite strings of integers in $[k]$.
We now define the hash function $\varphi:[3]^* \to \SL_n(\FF_p)$.
We start by choosing bijections \[s: [4] \to \{A^{\pm1},B^{\pm1}\},\quad s_{\lambda}: [3] \to \{A^{\pm1},B^{\pm1}\}\setminus s(\lambda)\] for each $\lambda \in [4]$.
	
Then, given \[x=(x_i)_{1\leq i\leq k} \in [3]^{k},\] we have the following inductive
definition:
	
\begin{itemize}
	\item $B_1 = s_1(x_1)$,
	\item $B_i = s_{\lambda}(x_i)$ with $\lambda=s\ii(B_{i-1}\ii)$, for $2\leq i \leq k$.
\end{itemize}
		
Finally, we set $\varphi(x) \vcentcolon = B_1 \cdots B_k\in \SL_n(\FF_p)$.
\end{definition}

\begin{remark}
Note that $G_{n,p}$ is $4$--regular, so that after the first digit $x_1$ of the input $x$,
there are exactly three non-backtracking edges in the graph by which to proceed.
The input $x$ can thus be viewed as encoding a reduced word in the free group $F_2$.
The lack of backtracking in the resulting walk on $G_{n,p}$ is crucial for the avoidance of
collisions, as well as for the reduction of mixing time.
\end{remark}

As stated in \cite{Arzhantseva2018Large}, the elements
$\{A,B\}$ generate a free subgroup of $\SL_n(\ZZ)$ and generate
$\SL_n(\FF_p)$ for all but finitely many values of $p$, and these facts give rise
to strong preimage and collision resistance of the resulting hash functions.

\subsection{A concrete example}\label{subsec: concrete example}

We finish this section by describing a family of
concrete examples of hash functions, which are constructed for the specific values
$a=4$, $b=2$ and $\ell=4$.
We do not know what minimal value of $n$ would ensure security.

\begin{definition}\label{definition: concrete example}

	Let $p$ be a prime, and let
	\[
A = \begin{pmatrix}
	1 & 4 & 0 & 0& \ldots & 0\\
	0 & 1 & 4 & 0& \ldots & 0\\
	0 & 0 & 1 & 4 & \ldots &0\\
	\vdots & & & & & \vdots\\
	&&&&\ldots&4\\
	0 & 0 & 0& 0&\ldots & 1\\
	\end{pmatrix}^4,\,
B = \begin{pmatrix}
	1 & 0 & 0 & \ldots & &0\\
	2 & 1 & 0 & \ldots & &0\\
	0 & 2 & 1 &  \ldots& &0\\
	0 & 0 & 2 & \ldots & &0\\
	\vdots &&&&\vdots&\vdots\\
		0& 0 & 0 & \ldots & {\strut \strut 2} & 1\\
	\end{pmatrix}^4\in \SL_{n}(\FF_p),
\]

	Let $s(1)=A$, $s(2)=B$, $s(3)=A^{-1}$, $s(4)=B^{-1}$.
	We define the functions $\{s_{\lambda}\}_{\lambda\in [4]}$ as follows:
	\begin{itemize}
		\item $s_1(1)=B$, $s_1(2)=A\ii$, $s_1(3)=B\ii$,
		\item $s_2(1)=A$, $s_2(2)=A\ii$, $s_2(3)=B\ii$,
		\item $s_{3}(1)=A$, $s_{3}(2)=B\ii$, $s_{3}(3)=B$,
		\item $s_{4}(1)=A$, $s_{4}(2)=A\ii$, $s_{4}(3)=B$,
	\end{itemize}

Given an input sequence $x = \{x_i\}_{i\in[1,k]} \in [3]^k$, we inductively define:

\begin{itemize}
	\item $B_1 = s_1(x_1)$
	\item $B_{i} = s_{\lambda}(x_{i})$, with $\lambda = s\ii(B_{i-1}\ii)$,
	for each $k \in [2,k]$.
\end{itemize}

Then, the sequence $x$ is hashed to the matrix:

	\[\varphi(x) = B_1 \cdots B_k .\]

	\end{definition}

Thus, we obtain a hash function for every $n\geq 3$.

\begin{example}\label{example: an instance of the hash function}
With $n=3$ we have:
	\[
A = \begin{pmatrix}
	1 & 16 & 96 \\
	0 & 1 & 16 \\
	0 & 0 & 1 \\
	\end{pmatrix},\,
B = \begin{pmatrix}
	1 & 0 & 0 \\
	8 & 1 & 0 \\
	24 & 8 & 1 \\
	\end{pmatrix}\in \SL_{3}(\FF_p),
\]

For example, if we consider the sequence $x=2232221$, following the procedure above we obtain the sequence:
\begin{itemize}
    \item $B_1=s_1(2)=A\ii$
    \item $B_2=s_1(2)=A\ii$, where we use the map $s_1$ because $B_1\ii=s(1)$,
    \item $B_3=s_1(3)=B\ii$, where we use the map $s_1$ because $B_2\ii=s(1)$,
    \item $B_4=s_2(2)=A\ii$, where we use the map $s_2$ because $B_3\ii=s(2)$,
    \item $B_5=s_1(2)=A\ii$, where we use the map $s_1$ because $B_4\ii=s(1)$,
    \item $B_6=s_1(2)=A\ii$, where we use the map $s_1$ because $B_5\ii=s(1)$,
    \item $B_7=s_1(1)=B$, where we use the map $s_1$ because $B_6\ii=s(1)$,
\end{itemize}

Finally, $x$ is mapped to $B_1B_2\ldots B_7$:
\[
\varphi(x)=A^{-2}B\ii A^{-3}B =
\begin{pmatrix}
694190977 & 233260720 & 29297952\\
-38379648 & -12896255 & -1619792\\
1191936 & 400512 & 50305\\
\end{pmatrix}
\in \SL_3(\FF_p).\]

We refer to Appendix \ref{section: implementation} for a Python/Sage implementation of this example.
\end{example}

\section{Properties of the constructed hash functions}\label{section:
properties hash function}

In this section we use graph and group-theoretic machinery to describe the security
of the hash functions defined above. We center our analysis on resistance to preimage
and collision breaking. The exposition is divided into five parts: first, we
establish a lower bound in the girth of the Cayley graphs of the group
$\SL_n(\FF_p)$ with respect to the generating system $\{A_p^{\pm 1},B_p^{\pm 1}\}$;
second, we describe the consequences of girth bounds for collision resistance;
third, we investigate mixing properties of the suggested platform.
Two last subsections are devoted to multiplicativity properties of the hash function, and showing that the so-called palindromic attack from \cite{grassl2011cryptanalysis} is inapplicable.

\subsection{Free groups and girth}\label{subsec: free girth}

The following proposition is in the spirit of \cite{Bromberg2017Navigating}.

\begin{proposition}\label{proposition: freeness implies big girth}
Let $A,B \in \SL_n(\ZZ)$ such that the entries of $A^{\pm 1}$ and $B^{\pm 1}$ are bounded in absolute value by a positive constant $c$.
If $A$ and $B$ generate a free subgroup of $\SL_n(\ZZ)$, then the girth of the Cayley graph of $\langle A_p,B_p \rangle \leq \SL_n(\FF_p)$, with respect to $\{A_p^{\pm 1},B_p^{\pm 1}\}$ is at least \[\left\lfloor \frac{\log (p-1)}{\log nc} \right\rfloor.\]
\end{proposition}

\begin{proof}
For any reduced word $w$ in $A^{\pm 1}$ and $B^{\pm 1}$, we write $w_{\ZZ}$ (resp. $w_{\FF_p}$) for the projection of $w$ to $\SL_n(\ZZ)$ (resp. $\SL_n(\FF_p)$).
It follows by a straightforward induction on $k$ that, if $w$ has length $k$, then the entries of $w_{\ZZ}$ cannot exceed $(nc)^k$ in absolute value.
Now, let $\ell$ be the girth of the corresponding Cayley graph.
Then, there exists a nontrivial reduced word $w$ of length
$\ell$ such that $w_{\FF_p} = 1$.
It follows that $w_{\ZZ}$ is of the form $1+pM$, where $M$ is an integer matrix.
Since $w$ is nontrivial and since $\{A,B\}$ generate a rank two free subgroup of
$\SL_n(\ZZ)$, the matrix $M$ is nonzero.
We conclude that $w_{\ZZ}$ has an entry of absolute value at least $p-1$.
Since the entries of $w_{\ZZ}$ cannot exceed $(nc)^k$ in absolute value,
we have that the length $\ell$ of $w$ is bounded below by
$\lfloor \frac{\log (p-1)}{\log nc} \rfloor$, and the desired conclusion holds.
\end{proof}

\subsection{Preimage and collision resistance, and
post-quantum heuristics}\label{subsection: resistance via polynomial equations}

    We now analyze the resistance of our model to finding preimages and to collisions.
    Observe that finding a preimage of a particular hash value (resp.~finding a
    collision of hash values) is equivalent
    to finding a factorization of a given group element
    (resp.~of the identity) in $\SL_n(\FF_p)$ with respect to the generating set. We note that in general, Tillich--Z\'emor hash
    functions seem to have robust collision resistance;
    see~\cite{Tinani23}.

    The matrices $A_p,B_p$ involved have order $p$, so a factorization can be seen as a family of equations $\{\eqref{equation: Em}\}_{m\ge0}$ with variables \[k_1,\ldots,k_m, \ell_1,\ldots,\ell_m \in\FF_p\] satisfying:
    \[\tag{$E_m$}\label{equation: Em} A^{k_1}B^{\ell_1}\ldots A^{k_m}B^{\ell_m} = M,\]
    for a given challenge $M \in \SL_n(\FF_p)$.
    This problem is equivalent to attacking the preimage resistance of the hash function.
    In the case where $M$ is the identity matrix, this is equivalent to attacking the collision resistance.
    Note that there are trivial solutions
    to preimage and collision breaking of the hash function, given that $A_p^p$ is the identity.
    Since the girth of
    $G_{n,p}$ is $O(\log p)$, we consider nontrivial solutions to
    preimage or collision breaking to be ones where
    \[C_1\log p\leq \sum_{i=1}^m (k_i+\ell_i)\leq C_2\log p,\] where $C_1$ and $C_2$ are
    positive constants depending on $n$ but not on $p$. Note that estimates for
    $C_1$ and $C_2$ can be produced, and that
    Proposition~\ref{proposition: freeness implies big girth} furnishes an estimate
    for $C_1$, for instance. Sharp values for $C_1$ and $C_2$ are of relatively
    minor consequence for us.

    Each entry of the left-hand side matrix in equation $\eqref{equation: Em}$
    is polynomial in \[k_1,\ldots,k_m,
    \ell_1,\ldots,\ell_m.\]
    
	Thus, the equation $\eqref{equation: Em}$, equivalent to attacking preimage of the hash function, corresponds
	to a system of $n^2$ multivariate polynomial equations over $\FF_p$.

    Solving multivariate polynomial equations over a finite field is known to be NP-hard \cite{GareyJ1979Computers}, which suggests a good level of security.
    Moreover, the reduction to solving multivariate polynomials, a class of hardness problems considered for standardization by the NIST, provides a certain degree of confidence that the hash function is post-quantum.
    We contrast this approach with schemes based on isogeny graphs,
    which reduce to a more well-defined problem, albeit one not known to be NP-hard.

    NP-hardness of a class of problems is a worst case complexity property, and
    for certain NP-hard classes of problems, relatively simple and efficient algorithms
    can find solutions in the vast majority of cases. Thus, NP-hardness of
    the underlying problem is not a guarantee of post-quantum behavior
    of the hash function.

    A more compelling case for the hash function to be post-quantum arises from
    empirical difficulty of factoring in special linear groups over finite fields.
    For instance, in~\cite{FaugerePPR11}, subexponential factorization algorithms
    were found for $\SL_2(\FF_{2^k})$, and these were only found in 2011.
    These algorithms rely essentially on the fact that the matrices are $2\times 2$,
    and on the fact that the underlying field has characteristic two. Thus, the methods
    do not generalize in any straightforward way to larger dimensional special linear
    groups nor to fields with odd characteristic. In practice, factoring matrices over
    finite fields is quite difficult, and implemented algorithms are inefficient.
    Hardness appears to be optimized when the system of equations resulting from
    $\eqref{equation: Em}$ is neither underdetermined nor overdetermined, i.e.~when
    the number of equations and variables is comparable. Thus, the larger
    the value of $p$ the more secure the hash function, at the expense of computational
    time and space, and the balance of degrees of freedom and constraints occurs when
    $n^2\sim\log p$, or in other words when $n$ is approximately the square root
    of the logarithm of $p$. This is precisely the balance to
    required so that the number of equations and number of
    variables are comparable, as per the foregoing discussion.
    We may then expect the factorization problem to take
    exponential time in the number of variables in this case.

\subsection{The mixing property}\label{subsection: mixing and indistinguishability}

By the mixing property, we mean that the output vertex of a random input --- in our case a random walk --- approaches the uniform distribution on the hash space.
When the random walk approaches the uniform distribution quickly,
mixing is observed even when the input messages have relatively small length, say polynomial in $\log p$.
More precisely, we have the following corollary of result of
Alon-Benjamini-Lubetzky-Sodin \cite{Alon2007NonBacktracking}, which characterizes
the rate at which a random walk on a graph converges to the uniform distribution
in terms of the spectral properties of its adjacency matrix:
\begin{theorem}
\cite[Theorem 1.1, cf.~proof of Theorem 1.3]{Alon2007NonBacktracking}
\label{theorem: mixing non backtracking}
	Suppose $d>2$.
	Let $X_0,X_1,\ldots,X_{\ell}$ be a non backtracking random walk on a
	$d$--regular connected graph $G$ with $N$ vertices. There is a constant $C>0$
	such that
	whenever $\ell \ge C\cdot \log N$ we have
	\[\left|\PP(X_{\ell}=v)-1/N\right| \le 1/N^2,\]
	for every vertex $v$ of $G$.
\end{theorem}

Examining the proof given in \cite{Alon2007NonBacktracking}, one finds that
the rate of mixing depends not so much on the graph $G$, but rather on the eigenvalues
of the adjacency matrix of $G$. Thus, if $G$ is a member of a
sequence $\{G_i\}_{i\in\mathbb N}$ of expander graphs, we may take the constant
$C$ in Theorem~\ref{theorem: mixing non backtracking} to depend only on the expansion
constant of the family.

It is well-known that mixing properties are desirable in Tillich-Z\'emor hashing schemes;
see \cite{charles2008cryptographic,TillichZemor08Collision}.
As explained in \cite{TillichZemor08Collision}, mixing is particularly
relevant when the hash functions are used in protocols whose security relies on the random oracle model; see \cite{BellareR1993RandomOracles} for example of such protocols.

The probability that an attacker finds a collision is at least the probability given by the birthday paradox, taking samples at random, which is minimized with the uniform probability \cite[Exercise 13.7]{Steele_2004}.

The relevance of this approach depends on the distribution of possible messages and in particular on how they are encoded, a question we do not address in the present paper.

Surprisingly few mathematical statements addressing the relationship between
mixing and attacks are present in the literature; an example can be found in
\cite[Theorem 3]{Sterner2021Commitment}, in the context of commitment schemes.
The following proposition is an immediate consequence of the previous theorem and the discussion above:

\begin{proposition}\label{prop: indistinguishability statement}
	Let $\varphi:[3]^k \to \SL_n(\FF_p)$ be the hash function of Definition \ref{definition: general hash function}, and let
	$N=|\SL_n(\FF_p)|$.
	Then, there is a positive constant $C$ such that, if $k \ge C\cdot \log N$, and $m$ is taken uniformly at random in $[3]^k$, then we have
	\[\left|\PP(\varphi(m)=M)-1/N\right| \le 1/N^2,\]
	for every $M \in \SL_n(\FF_p)$.
\end{proposition}

\subsection{Multiplicativity and parallel computing}\label{subsection: homomorphism parallel computing}

The hash functions considered in this article take as input a string of
integers in $[3]$, convert each integer into a matrix of the form
$\{ A^{\pm1}, B^{\pm1} \}$, and finally output the product of these matrices.

The fact that we require the underlying walk to be non-backtracking
implies that this mapping is not locally determined: a given digit in the
string is mapped to a matrix that depends on the previous digits.
This dependency can be dramatic:
for example, according to Definition \ref{definition: concrete example} a sequence of the form $133\cdots3$ will be mapped to the product $B\cdot B\cdot B \cdots B$,
while a sequence of the form $333\ldots3$ will be mapped to the product
$B\ii\cdot B\ii\cdot B\ii \cdots B\ii$.
In particular, the last digit $3$ of these two strings can be mapped to
different matrices, depending on the first digit in the string. The endpoints of the
corresponding walk in the Cayley graph may be far away from each other.

As a consequence, the function $\varphi$ need not be multiplicative under
concatenation of strings, which is generally a desirable feature for hash function.
This lack of multiplicativity makes it difficult to perform parallel computations with
the given hash functions, as we now investigate in more detail.

\subsubsection{Good and bad tails}

Recall that for a finite set $X$, the notation $X^*$ is used for
the set of finite length strings of elements of $X$.
As before, the notation $[3]$ denotes the set $\{1,2,3\}$.

\begin{definition}
	Let $G$ be a finite group, generated by two elements $A$ and $B$.
	Let $\tilde\varphi: [3]^* \to \{A^{\pm 1}, B^{\pm 1}\}^*$.

	A string $s \in [3]^*$ is called a \emph{good tail} with respect to
	$\tilde\varphi$ if there exists $S \in \{A^{\pm 1}, B^{\pm 1}\}$ such that for every $s' \in [3]^*$, the last letter of $\tilde\varphi(s's)$ is $S$,
	where here $s's$ is the string obtained from the concatenation of $s'$ and $s$.
	A string which is not a good tail is called a \emph{bad tail}.	
\end{definition}

Local constraints in Definition \ref{definition: concrete example}, can be obtained by the following fact:

\begin{fact}
The function \[\tilde\varphi: \{x_i\} \in [3]^* \mapsto \{B_i\} \in \{A^{\pm 1}, B^{\pm 1}\}^* \] constructed in Definition \ref{definition: concrete example} has the following good tails: 11, 31, 22, 32, 13 and 23.
\end{fact}
\begin{proof}
It is straightforward to check that:

\begin{itemize}
	\item any string ending in 11 or 31 outputs a string ending in $A$;
	\item any string ending in 22 or 32 outputs a string ending in $A\ii$;
	\item any string ending in 13 outputs a string ending in $B$;
	\item any string ending in 23 outputs a string ending in $B\ii$.
\qedhere
\end{itemize}

\end{proof}


The following proposition shows that the mapping above is optimal.

\begin{proposition}\label{proposition: max six good tails}
Special linear group based hash functions (Definition \ref{definition: general hash function}) admit at most six good tails of length two.
\end{proposition}

	The bound in Proposition~\ref{proposition: max six good tails} is
	sharp, as shown via
	the mappings from Definition \ref{definition: concrete example}.
	The proof of
	Proposition \ref{proposition: max six good tails} will follow from the
	following lemma:

	\begin{lemma}\label{lemma: always at least one bad tail per digit}
	Let \[\tilde\varphi: \{x_i\} \in [3]^* \mapsto \{B_i\} \in \{A^{\pm 1}, B^{\pm 1}\}^*
	\] be a special linear group based hash function (Definition \ref{definition: general hash function}), and	
	let $b \in [3]$. Then, there exists $b' \in [3]$ such that $b'b$ is a bad tail with respect to $\tilde\varphi$.
	\end{lemma}
	
	\begin{proof}
	The only freedom that we have in the construction of Definition \ref{definition: general hash function} is how we define the
maps $s_i$.
	We call elements of $\{A^{\pm1},B^{\pm1}\}$ \emph{step matrices}.
	Using $\tilde\varphi$, we say that the elements
	of $[3]$ are \emph{encoded} by step matrices.
	We summarize the definitions of the maps $s_i$ in a table,
	with one row for each step matrix, and one column for each element of $[3]$.
	Each cell from this tabular contains a step matrix.
	
	To use this table,
	start from a string $\{x_i\} \in [3]^*$.
	Say that for some $i>1$, we want to find the step matrix associated with $x_i$.
	Let $S$ be the step matrix encoding $x_{i-1}$.
	Then, the step matrix encoding $x_i$ is the step matrix in the cell located in the row labelled by $S$ and in the column labelled by $x_i$.
	
	It follows from the definition of the maps $s_i$ that for each step matrix $S$, the row labelled by $S$ contains exactly the three matrices in the
	set $\{A^{\pm1},B^{\pm1}\} \setminus S^{-1}$.
	The mapping from Definition \ref{definition: concrete example} can be described as in Figure~\ref{figure: mappings of the si's}.
	
	\begin{figure}\caption{Description of the maps $s_i$
	in Definition \ref{definition: concrete example}}\label{figure: mappings of the si's}
	\begin{center}
		\begin{tabular}{ r l | lll }
     \multicolumn{2}{l|}{last step matrix} & 1 & 2 & 3 \\ \hline
     & $A\ii$ & $B$ & $A\ii$ & $B\ii$ \\
     & $B\ii$ & $A$ & $A\ii$ & $B\ii$ \\
     & $A$ & $A$ & $B\ii$ & $B$ \\
     & $B$ & $A$ & $A\ii$ & $B$
		\end{tabular}
	\end{center}
	\end{figure}

	Moreover, since each matrix can actually be the last step matrix used, every cell of the table can potentially be used.
	Fix an element $b \in [3]$, and assume for a
    contradiction that every integer $b' \in [3]$ has the property
    that $b'b$ is a good tail.
	The column corresponding to each $b'$ has to contain at least two different step matrices.
	This implies that, in the row labelled by $b$, at least two cells contain the same step matrix.
	Since this is true for each $b' \in [3]$, this implies in particular that there is a step matrix $S$ that is contained three times in the column labelled by $b$.
	The fourth cell of this column cannot be part of the row labelled by $S$ since this would give rise to another $S$ in a different column.
	This implies that the label $S'$ of this row appears in a cell of another column.
	Additionally, this column contains a cell with step matrix not equal to $S'$.
	Then, the label $b' \in [3]$ of this column gives us an integer
	having the property that inputs ending by $b'b$ can have either $S$ or another matrix as a final matrix.
	This is a contradiction and concludes the proof of the lemma.	
	\end{proof}
	
	\begin{proof}[Proof of Proposition \ref{proposition: max six good tails}]
	From Lemma \ref{lemma: always at least one bad tail per digit}, to each integer of $[3]$ corresponds at least one bad tail, giving three different bad tails.	
\end{proof}

As remarked previously, Definition \ref{definition: concrete example} shows that the estimate in Proposition~\ref{proposition: max six good tails} is sharp,
	and so that in some sense, we have found an optimal way of defining the maps $s_i$.

%
%

\subsubsection{Multiplicativity}
It follows from the discussion
of good and bad tails above that multiplicativity of the hash function
can be obtained by restricting to sequences whose product ends with the matrix $s(1)\ii$.

\begin{fact}
	In Definition \ref{definition: concrete example}, we have $\varphi(s_1 * s_2) = \varphi(s_1) \cdot \varphi(s_2)$, provided $s_1$ ends with $22$ or $32$.
\end{fact}

\subsubsection{Parallel computing}
Multiplicativity of the hash function under suitable conditions can be leveraged to
compute its values by parallel computation.
First, look for good tail substrings, namely: 11, 31, 22, 32, 13 or 23.
For generic messages, one would expect such substrings to be quite common.
Next, split the input immediately following one of these strings,
and apply a slightly modified hash function (i.e.~using the relevant $s_i$ instead of
$s_1$ in the first matrix mapping).
Finally, compute the product of the hash outputs.

\begin{example}
Say we want to hash the string $1321321323$. Observe that: \break$132\mathbf{13}21323$.
We thus compute: $M_1 = \varphi(13213)$ and $M_2 = \varphi'(21323)$, where
$\varphi'$ is defined analogously to $\varphi$ in Definition
\ref{definition: concrete example}, apart from the fact that $B_1$ is set to be
$s_4(x_1)$ instead of $s_1(x_1)$, since $\varphi(13213)$ is a product ending by $B$.
Finally,
$\varphi(1321221323) = M_1 \cdot M_2$.
\end{example}

\subsection{Palindromic attacks}
One of several proposals of hashing by walks on Cayley graphs can be found in
\cite{tillich1994hashing}, wherein the Cayley graph is that of $SL_2(\FF_{2^n})$.
A method for finding collisions for this hash function is presented in
\cite{grassl2011cryptanalysis} (cf.~\cite{FaugerePPR11});
we argue that the attack does not apply
in our case, though our evidence for this claim is primarily empirical.

The idea of \cite{grassl2011cryptanalysis} is to find collisions on
\textit{palindromes}; that is, digit-string entries that are invariant under
reversing the order.
To begin, one conjugates generators of $SL_2(\FF_{2^n})$ to obtain new generators
which give rise to an isomorphic graph, but which are symmetric matrices.
That is, if the original generators are $\{A^{\pm1}, B^{\pm1}\}$,
one finds a matrix $C$ such that $\hat{A}=CAC^{-1}$ and $\hat{B}=CBC^{-1}$ are both symmetric matrices.

We first note that in our case, finding $C$ is not easy; for $SL_3(5)$ and $SL_3(7)$, about $0.02\%$ of the elements satisfy this criterion.
Moreover, there is no obvious way to compute $C$; attempts to calculate the entries
of such a matrix directly have proved resistant to equation solving methods in
standard computer algebra systems - indeed, this approach is actually less
efficient than just checking all possible matrices. Therefore,
we do not have much data for larger primes, since the naive method used to find
a suitable matrix $C$ quickly becomes computationally infeasible.

Provided one can find a matrix $C$, it follows that collisions in the hash
function with respect to $\hat{A},\hat{B}$ as generators are exactly the collisions with $A, B$ as generators; one can therefore
rename the matrices $\hat{A}, \hat{B}$ as $A, B$. One then proceeds according
\cite[Lemma 1]{grassl2011cryptanalysis}: upon input of a palindromic string $v$, the output of the product of conjugated generators in $SL_2(\FF_{2^n})$ will always be a symmetric matrix.

Since our hash function requires avoidance of backtracking in the walk,
we are not guaranteed a palindromic matrix product from a palindromic input string; however, since one could reverse-engineer the necessary input to obtain a palindromic matrix product, we proceed to discuss palindromic matrix products without reference to
hash function.

It turns out, as one may check easily by induction, that a palindromic product in
symmetric generators will itself be symmetric. The ultimate goal of
\cite{grassl2011cryptanalysis} is to use this fact to demonstrate that
the function \[\rho:M\mapsto AMA+BMB,\] where $M$ is a palindromic product,
outputs a matrix populated with either zeroes or the square of a field element
appearing as an entry in $M$.
One then employs number theoretic tricks to force the nonzero elements to $0$ in
$M$ and thus to obtain $\rho(M)=0$.
One thus builds distinct palindromic decompositions of the same matrix.

Consider the generators from Definition~\ref{definition: concrete example} over $\SL_3(\FF_{11})$. Transforming these generators with respect to the matrix
\[C =
\begin{pmatrix}
2 & 6 & 10 \\
5 & 3 & 10 \\
2 & 3 & 3
\end{pmatrix},
\]
one checks that the palindrome $M=ABABA$ is such that
\[M=
\begin{pmatrix}
7 & 4 & 2 \\
4 & 0 & 6 \\
2 & 6 & 6
\end{pmatrix},
\quad
\rho(M)=
\begin{pmatrix}
2 & 1 & 5 \\
1 & 1 & 7 \\
5 & 7 & 7
\end{pmatrix}.
\]

In particular, for each $i\in[10]$, the matrix $\rho(M)$ contains an entry
that is not the $i^{th}$ power of any entry of $M$. This furnishes evidence that
for $p=11$, there is little hope of extending \cite[Corollary 1]{grassl2011cryptanalysis} to our context; we argue that the lack of closed form of transformed generators in general,
the difficulty of finding them for larger parameters, and this example with a small
value of $p$, conspire to provide strong evidence that the approach will fail in general.

\section{Conclusions}

\label{section: conclusions}
We have presented new Tillich-Z\'emor hash functions, with platforms Cayley graphs of $\SL_n(\FF_p)$ for $n \ge 3$.
We show that choosing appropriate generating matrices produces graphs without small cycles, and having a quick mixing property, both of which are highly desirable for preimage and collision resistance. Moreover, flexibility of choice of generating matrices and of the dimension $n$ gives the option of increasing the complexity of attacks.
Future work includes the exact computation of the spectral gap and the prime $p_0$ (cf.~item (i) of Theorem \ref{theorem: ArzBis}). Moreover, simulations
should be carried out in order to compare with other existing schemes
and determine the optimal values of $p$ and $n$ to be taken in implementations.

\section*{Acknowledgments}
The authors would like to thank Ludovic Perret, Hadi \break Salmasian, Vladimir Shpilrain, and Bianca Sosnovski for helpful comments and discussions.
The authors thank ITS for hosting CL, CB, and RF while this project has been done partially.
DK thanks Institut des Hautes \'Etudes Scientifiques - IHES for providing stimulating environment while this project was partially done.
She has conducted this work partially with the support of ONR Grant 62909-24-1-2002.

The authors are also grateful to two anonymous referees who gave interesting suggestions on how to improve the paper.

\newpage
\appendix
\section{Implementation of the hash function}\label{section: implementation}
In this appendix, we provide an implementation of the hash function given in Example \ref{example: an instance of the hash function}.

\begin{center}
\begin{minipage}{0.98\textwidth}
\begin{python}
from sage.all import *
# a 1024 bits random prime p
nbits = 1024
p = random_prime(2**nbits, lbound = 2**(nbits-1))
assert is_prime(p)

# definition of the matrices A, B and their inverses
A = matrix(GF(p),[[1, 16, 96], [0, 1, 16], [0, 0, 1]])
Ainv = A**(-1)
B = matrix(GF(p),[[1, 0, 0], [8, 1, 0], [24, 8, 1]])
Binv = B**(-1)

# identification between digits (in base 4) and matrices
s = [A, B, Ainv, Binv]
sigma = [[B, Ainv, Binv], [A, Ainv, Binv], [A, Binv, B], [A, Ainv, B]]

# definition of the hash function
def hash(string):
    # input must be a string
    # the only characters allowed are '1', '2' and '3'
    out = identity_matrix(GF(p), 3)
    inv_prec = A
    for k in range(len(string)):
        i = s.index(inv_prec)
        step_matrix = sigma[i][int(string[k])-1]
        inv_prec = step_matrix**(-1)
        out = out * step_matrix
    return out

# test of the example given in the paper
string = '2232221'
assert hash(string) == matrix(GF(p),
[[694190977, 233260720, 29297952],
[-38379648, -12896255, -1619792],
[1191936, 400512, 50305]])
\end{python}
\end{minipage}
\end{center}








\newcommand{\etalchar}[1]{$^{#1}$}

\end{document}